\documentclass[11pt]{article}
\usepackage{amsmath}
\usepackage{amssymb}
\usepackage{amsmath}
\usepackage{amstext}
\usepackage{amsthm}
\usepackage{amsbsy}
\usepackage{amsfonts}
\usepackage{dsfont}
\usepackage{marvosym}
\usepackage{mathtools}
\usepackage{authblk}
\usepackage[symbol]{footmisc}
\usepackage[nottoc]{tocbibind}
\usepackage[margin=3cm]{geometry}
\usepackage{bbm}
\usepackage[symbol]{footmisc}

\usepackage{float}

\newtheorem{theorem}{Theorem}[section]
\newtheorem{lemma}[theorem]{Lemma}

\newtheorem{observation}[theorem]{Observation}

\newtheorem{definition}[theorem]{Definition}

\newcommand{\cT}{{\cal T}}
\newcommand{\cH}{{\cal H}}

\newcommand{\cC}{{\cal C}}
\newcommand{\cI}{{\cal I}}
\newcommand{\cX}{{\cal X}}

\newcommand{\cn}{k} 
\newcommand{\gam}{M} 

\newcommand{\skp}{\tilde{z}}

\begin{document}

\begin{center}
{\Large An Efficient Reduction of a Gammoid to a Partition Matroid\footnote[1]{\textbf{Funding} \textit{Marilena Leichter:} Supported in part by the Alexander von Humboldt Foundation with funds from the German Federal Ministry of Education and Research (BMBF) and by the Deutsche Forschungsgemeinschaft (DFG), GRK 2201.\\
\textit{Benjamin Moseley:} Supported in part by a Google Research Award, an Infor Research Award, a Carnegie Bosch Junior Faculty Chair and NSF grants CCF-1824303,  CCF-1845146, CCF-1733873 and CMMI-1938909.\\
\textit{Kirk Pruhs:} Supported in part  by NSF grants  CCF-1535755, CCF-1907673,  CCF-2036077 and an IBM Faculty Award.\\
\textit{Email addresses:} \texttt{marilena.leichter@tum.de} (Marilena Leichter), \texttt{moseleyb@andrew.cmu.edu} (Benjamin Moseley), \texttt{kirk@cs.pitt.edu} (Kirk Pruhs)}}\\
\par\vspace{\baselineskip}

Marilena Leichter\textsuperscript{1}, Benjamin Moseley\textsuperscript{2}, Kirk Pruhs\textsuperscript{3}
\par\vspace{\baselineskip}
\footnotesize{\textsuperscript{1} Department of Mathematics, Technical University of Munich, Germany\\
\textsuperscript{2}Tepper School of Business, Carnegie Mellon University, USA\\
\textsuperscript{3} Computer Science Department, University of Pittsburgh, USA}
\end{center}
\par\vspace{\baselineskip}

\begin{abstract}
Our main contribution is a polynomial-time algorithm
to reduce a $k$-colorable gammoid to a 
$(2k-2)$-colorable partition matroid.
It is known that there are gammoids that can not be reduced to
any $(2k-3)$-colorable partition matroid, so this result is tight. 
We then discuss how such a reduction can be used to obtain
polynomial-time algorithms with better approximation ratios
for various natural problems related to coloring and list
coloring the intersection of matroids. 
\end{abstract}

\section{Introduction}

Our main contribution is a polynomial-time algorithm
to reduce a $k$-colorable gammoid to a 
$(2k-2)$-colorable partition matroid.
Before elaborating on the statement of this result, we first give the necessary 
definitions, and the most relevant prior work. 
After stating the result we then explain some of the
algorithmic ramifications. 

\subsection{Definitions}

 A set system is a pair $\gam = (S,\cI)$ where  $S$ is a universe of $n$ elements and $\cI\subseteq 2^S$ is a collection of subsets of $S$. 
  Sets in $\cI$ are called \textbf{independent} and the
  rank $r$ is the maximum cardinality of a set in $\cI$. 
  A partition $C_1, C_2, \ldots, C_{k}$ of $S$ into independent sets is a 
  $k$-coloring of $\gam$. 
  The \textbf{coloring number} of  $\gam$ is the smallest $k$ such that
  a $k$-coloring exists.
    
 If each element $s \in S$ has an associated list $A_s$ of allowable colors, 
 then a list  coloring is a coloring $C_1, C_2, \ldots, C_{j}$  such
 that if an $s \in S$ is in $C_i$ then $i \in A_s$. 
 The \textbf{list coloring number} of  $\gam$ is the smallest $k$ 
 that guarantees that if for $s \in S$ it is the case that if $|A_s| \ge k$
 then a list coloring exists. 
 
 If $R \subseteq S$ then the restriction of $M$ to $R$,
 denoted by $M \mid R$, is a set system where the universe
 is $S \cap R$ and where a set $I \subseteq S$ is independent if
 and only if $I \subseteq R$ and $I \in \cI$. 

  A hereditary set system is a set system where   if 
  $A \subseteq B \subseteq S$ and $B \in \cI$ then $A \in \cI$.
  A matroid is an hereditary set system with the additional properties that $\emptyset \in \mathcal{I}$ and if 
   $A  \in \cI $, $B \in \cI$, and $|A |  < |B|$ then
   there exists an $s \in B \setminus A$ such that $A \cup \{s\} \in \cI$.
   The intersection of matroids $(S, \cI_1), \ldots ,(S, \cI_\ell)$ on common universe 
   is a hereditary set system with universe $S$ where a set $I \subseteq S$ is 
   independent if and only if for all $i \in [1, \ell]$ it is the case
   that $I \in \cI_i$.

  A \textbf{gammoid} is a matroid that has a graphical representation
  $(D=(V, E), S, Z)$, where $D=(V,E)$ is a directed graph,
  $S \subseteq V$ is a collection of source vertices and $Z \subseteq V$ is a collection
  of sink vertices. 
   In the gammoid that is represented by $D$ a set $I\subseteq S$ is in $\cI$ \emph{if and only if} there exists $|I|$ vertex-disjoint paths from 
   the vertices in $I$ to some subcollection of vertices in $Z$.
 A \textbf{partition matroid} is a type of matroid that can be represented
 by a partition $\cX$ of $S$.
 In the partition matroid that is represented by the partition $\cX$
 a set $Y\subseteq S$ is in $\cI$ \emph{if and only if} $|Y \cap X| \leq 1$ for all $X \in \cX$.~\footnote{Technically one can generalize this to let 
 there be a separate upper bound for each $X_i$ on the number of elements
 $Y$ can obtain from $X_i$, but in this paper we only consider partition
 matroids where the bound is one.}

 A matroid $N$ is a reduction (also called weak map) of a matroid $M$, with the same
 universe, if and only if every independent set in $N$ is also
 an independent set in $M$. If $N$ is a partition matroid,
 then we say that there exists a \textbf{partition reduction} from $M$ to $N$. 
\cite{im2020matroid} defined the following decomposability 
concept, which generalizes partition reduction.

\begin{definition}
   A matroid $\gam = (S, \mathcal{I})$ is \textbf{$(b,c)$-decomposable} if $S$ can be partitioned into sets $X_1, X_2,\dots, X_\ell$
     such that:
     \begin{itemize}
     \item  For all $i \in [\ell]$, it is the case that $|X_i| \leq c \cdot \cn$, where $k$ is the coloring number of $\gam$.
     \item For a set $Y =\{v_1, \ldots, v_{\ell}\}$, consisting of one representative element $v_i$ from
         each
 $X_i$, the matroid $M \mid Y$ is $b$ colorable. 
     \end{itemize}
\end{definition}
If $b=1$ then $X_1, X_2,\dots,X_\ell$ represents a partition matroid.
Thus $(1, c)$-decomposability means there exists a 
partition reduction  where the coloring number increases by
at most a factor of $c$.

 \subsection{Prior Work}

There are two prior, independent, papers in the literature that are directly relevant 
to our results. 
\cite{berczi2019list} showed that any gammoid $\gam$ admits a $(1, (2- \frac{2}{\cn}))$-decomposition. This proof is constructive, and can be converted into an algorithm. The resulting algorithm is essentially
a local search algorithm that selects a neighboring solution in 
the dual matroid in such
a way that an auxiliary potential function always decreases. 
But there seems to be little hope
of getting a better than exponential bound on the time, at least using techniques
from \cite{berczi2019list} as the potential 
can be exponentially large.  
Further \cite{berczi2019list} shows that no better bound is achievable.

 \cite{im2020matroid} gave a polynomial-time
 algorithm to construct a $(18,1)$-decomposition of a gammoid.   
 The reduction  was shown by leveraging prior work on unsplittable flows~\cite{Kleinberg96}. Both paper \cite{berczi2019list,im2020matroid} also observed that partition reductions are relatively easily obtainable for other common types of combinatorial matroids. In particular transversal matroids are $(1,1)$-decomposable \cite{berczi2019list, im2020matroid}, graphic matroids are $(1,2)$-decomposable \cite{berczi2019list,im2020matroid} and paving matroids are $(1,\lceil \frac{r}{r-1}\rceil)$-decomposable if they are of rank $r$ \cite{berczi2019list}. 

The main algorithmic result from \cite{im2020matroid} is:

\begin{theorem} \cite{im2020matroid}
\label{thm:immain}
Consider  matroids $M_1, M_2, \ldots, M_\ell$ defined over a common universe, 
where matroid $M_i$ has coloring number $k_i$.
There is a polynomial-time algorithm that, given  a $(b_i, c_i)$-decomposition
of each matroid $M_i$, computes 
a coloring of the intersection of $M_1, M_2, \ldots, M_\ell$ using  at most $\left(\prod_{i \in [k]} b_i \right) \cdot  \left( \sum_{i \in [k]} c_i \right)  k^*$ colors, where $k^*= \max_{i \in [l]} k_{i}$.
 \end{theorem}

Combining Theorem \ref{thm:immain} with the decompositions in \cite{berczi2019list,im2020matroid}
one obtains $O(1)$-approximation algorithms for problems that can be expressed
as coloring problems on the intersection of $O(1)$ common combinatorial matroids.
Several natural examples of such problems are given in \cite{im2020matroid}.
\cite{Aharoni06theintersection} showed that for two matroids $M_1$ and
$M_2$ over a common universe with coloring numbers $k_1$ and $k_2$, the coloring number $k$ of $M_1 \cap M_2$ is at most $2 \max(k_1, k_2)$.
The proof in \cite{Aharoni06theintersection}  uses topological arguments  
that do  not  directly  give  an  algorithm  for  finding the coloring. The list coloring number of a single matroid equals its coloring number \cite{seymour1998note}. For the intersection of two matroids \cite{berczi2019list} observed that a list coloring could be efficiently computed by partition reducing each of the matroids. A consequence of the results
in \cite{berczi2019list} is a constructive
proof that the list coloring number of $M_1 \cap M_2$ is at most $2 \max(k_1, k_2)$ if $M_1$ and $M_2$ are each one of the standard
combinatorial matroids. Further a consequence of the results
in \cite{im2020matroid} is an efficient algorithm to compute such a list 
coloring  if $M_1$ and $M_2$ are each one of the standard
combinatorial matroids besides a gammoid. 

For hereditary set systems the coloring number is equal to the  set cover number.
  Set cover has been studied extensively in the field of approximation algorithms.  The greedy algorithm has an approximation ratio of $H_n \approx \ln n$ and this is essentially optimal assuming $P \neq NP$ \cite{Vazirani2001,Williamson2011}.

\subsection{Our Main Result and Its Algorithmic Applications}

We are now ready to state our main result:

\begin{theorem}\label{thm:main}
A partition reduction from a $k$-colorable gammoid to a $(2k-2)$-colorable partition matroid can be computed in polynomial time given a directed graph $D$ that represents $\gam$ as input.  
\end{theorem}

Recall that \cite{berczi2019list} showed that the $(2k-2)$ bound is tight.

Combining our main result, Theorem \ref{thm:main}, with Theorem \ref{thm:immain} from
\cite{im2020matroid} we  obtain significantly better approximation guarantees for matroid
intersection coloring problems in which one of the matroids is a gammoid. 
One example is given by the following problem. Initially assume
that the input consists
of a directed graph $D$ with a designated file server location (a sink) 
and a collection of clients requesting files from the server
at various locations in the networks (the sources). 
The goal is to as quickly as possible get every client the file that they
want from the server, where in each 
time step one can service any collection of clients for which there exist disjoint
paths to the server. This is a matroid coloring problem that can be solved
exactly in polynomial-time ~\cite{Edmonds65}.
Now assume that additionally the input identifies the company to which each
client is employed by, and for each company there is a Service Level Agreement (SLA)
that upper bounds on how many 
clients from that company can be serviced in one time unit.
Now, the problem becomes a matroid intersection coloring problem,
where the intersecting matroids are a gammoid and a partition matroid. 
Using the $(18, 1)$-decomposition of a gammoid 
and Theorem \ref{thm:immain} from \cite{im2020matroid}
one obtains a polynomial-time $36$-approximation algorithm.
However, combining the $(2k-2)$-partition reduction of a gammoid from Theorem \ref{thm:main}
with Theorem \ref{thm:immain} from \cite{im2020matroid} we now obtain
a polynomial-time 3-approximation algorithm.\footnote{This is because Theorem \ref{thm:immain} is a $(1,2)$-decomposition of a gammoid and a partition matroid is in itself a $(1,1)$-decomposition, so Theorem \cite{im2020matroid} states the approximation is 3.}

Another algorithmic consequence is an efficient algorithm to list 
coloring the intersection $M_1 \cap M_2$ of a $k_1$-colorable matroid
$M_1$ and a $k_2$-colorable matroid $M_2$ if 
the list of allowable colors for each element has cardinality at least
$2 \max(k_1, k_2)$, and each of the matroids is either a graphic matroid, paving matroid, transversal matroid, or gammoid. Casting this into the context our running
file server example implies that additionally each client has a list of allowable
times when the file transfer may be scheduled. Our partition reduction
of a gammoid then yields an efficient algorithm to find a feasible
schedule as long as the cardinality of allowable times for each
client is at least $2\max(k_1, k_2)$, where $k_1$ is time required
 if the network had infinite capacity (so only the
SLA constraints come into play), and $k_2$ is the time required if the
SLA allowed infinitely many file transfers  (so only the network capacity constraints come into play).

Set cover is a cannonical algorithmic problem. So
there is considerable interest in discovering examples of natural special types of set cover instances where $o(\log n)$  approximation is possible.
For example, several geometrically based types of instances are known, 
for example covering points in the plane using a discs \cite{mustafa2009ptas}, where  a polynomial time approximation scheme is known.   
Our results provide another example of such a natural special case, namely when the sets come from the intersection of a small number of 
standard combinatorial matroids.

Theorem \ref{thm:main} and its proof   reveal structural properties of gammoids that would seem likely to be of use to address future research on gammoids.

\subsection{Overview of Techniques}
\label{sec:tech}

Given a graphic representation of a gammoid, an optimal coloring  
can be computed in polynomial-time~\cite{Edmonds65}. 
By superimposing the source-sink paths for the various color classes
one can obtain a flow $f$ from the sources to the sinks that moves at most
$k$ units of flow over any vertex. Using standard cycle-canceling 
techniques~\cite{ahuja1993network} one can then convert $f$ to what we call an acyclic flow.
A flow $f$ is acyclic  if for every undirected cycle $C$ in $D$ 
 at least one edge in $C$ 
either has flow $k$ 
or has no flow.
Thus by deleting edges that support no flow in $f$, as they are unnecessary,
we are left with a forest $\cT$ of edges that have flow in the range $[1, k-1]$
and a collection of disjoint paths, which we call highways, that have flow $k$. 
 See Figure \ref{fig:trees}.

\begin{figure}[t] 
 \centering
     \includegraphics[width=0.8\textwidth]{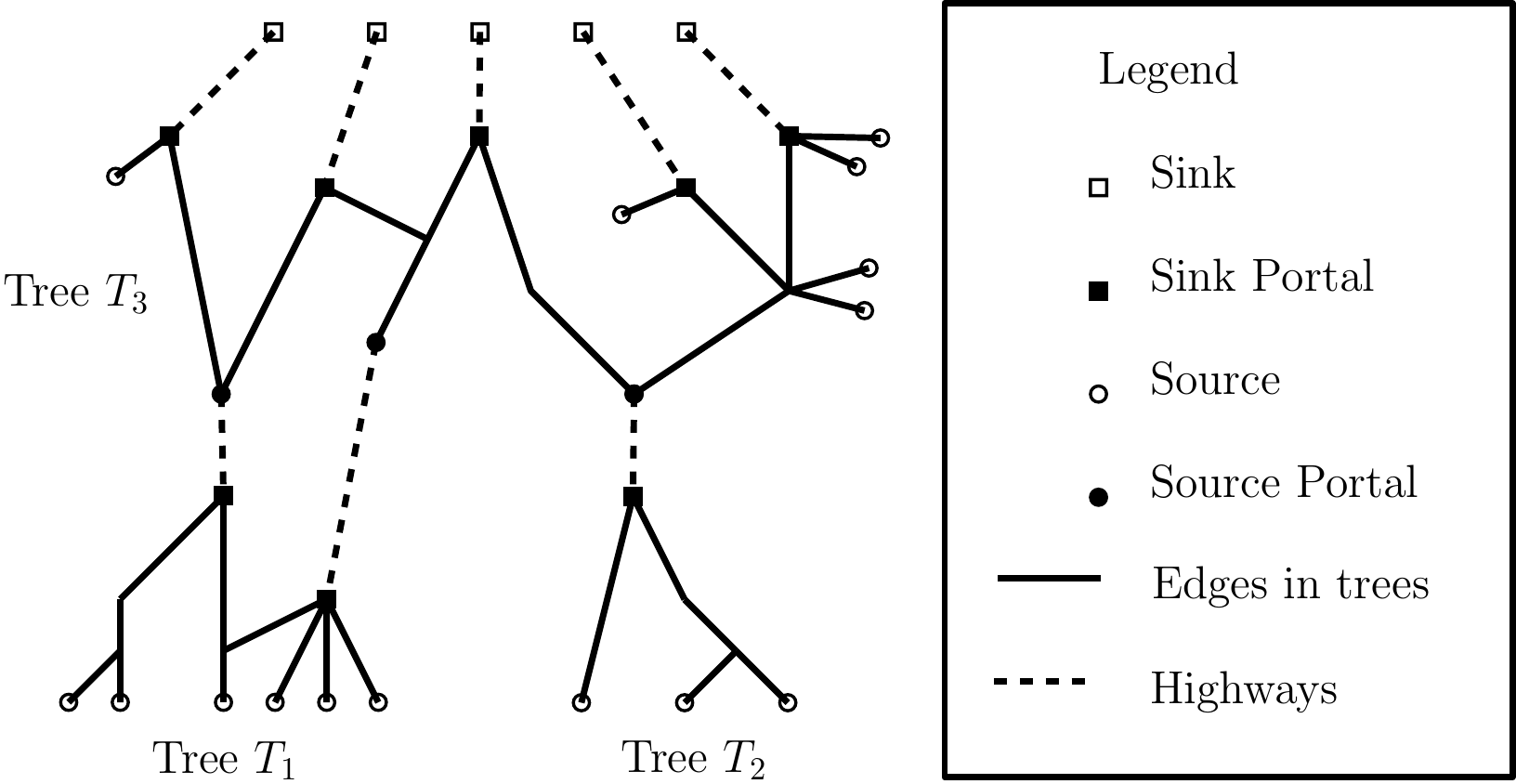}
     \caption{Example of trees created.  Here $k=3$. Source portals are matched to  sink portals along a path not in the trees.  All sink portals will have $k$ units of flow entering them and source portals have $k$ units leaving. \label{fig:trees}}
\end{figure}

Now each part $X$ in the computed partition $\cX$ will be entirely contained in one 
tree $T \in \cT$, and the parts $X$ in a tree $T \in \cT$ are computed
independently of other trees in $\cT$. 
There can be four types of vertices in $T$:
(1) sources $s$ that have outflow 1, (2) source portals $\tilde s$, which are vertices that have a highway directed into them, and which have outflow $k$ in $T$, (3) sink portals $\tilde z$, which are vertices that have a highway directed out of them, and which have inflow $k$ in $T$, and (4) normal vertices. Again see Figure \ref{fig:trees}.

We give a recursive partitioning algorithm for forming the parts $X$ in a tree $T \in \cT$.
On each recursive step our algorithm first identifies a single part $X$ of
at most $2k-2$ sources and an associated sink portal that are in some sense near
each other on 
the edge of $T$. The algorithm then removes these sources and sink portal
from $T$,  and reconnects disconnected sources back into appropriate places in
$T$. The algorithm then recurses on this new tree $T$.

Most of the proof that our partitioning algorithm produces a 
$(1, 2 - \frac{2}{k})$-decomposition focuses on routing individual trees in $\cT$. 
So let  $Y$ be a collection of sources such that for all $X \in \cX$
$| Y \cap X| \le 1$. 

The first key part  is proving 
 that as the partitioning algorithm recurses on a tree $T$, it is always possible to route both
the flow coming into $T$,  and the flowing emanating within $T$,    out of $T$, without routing more than $k$ units of
flow through any vertex in $T$. Note that as the algorithm recurses
the tree $T$ loses a sink portal (which reduces the capacity of the flow that can leave $T$ by $k$)
and loses up to $2k-2$ sources (which means there is less flow emanating in $T$ that
has to be routed out). 

The second key part is to prove
that there is a vertex-disjoint routing from the source portals in $T$ and the sources in $T \cap Y$
to the sink portals in $T$.
To accomplish we trace our partitioning algorithm's recursion backwards.
So in each step a new collection $X$ of sources and a sink portal 
is added back into $T$. We then prove by induction that no matter how
the previously considered sources in $ Y$ were routed, there is always a feasible way to 
route the chosen source in $Y \cap X$ to a sink portal in $T$. 
Then we finish by observing that unioning the routings constructed within the trees with the highways
gives a feasible routing for $Y$.

\section{Preliminaries}
This section introduces notation and other necessary definitions. Let $D=(V,E)$ be a directed graph that represents   a gammoid.  Let $S \subseteq V$ be a set of sources, and 
$Z \subseteq V$ be the collection of sinks.  We may assume without loss
of generality that:

\begin{itemize}
    \item Each vertex $v \in V$ has either out-degree 1
    or in-degree 1. 
    \item
    Each source $s \in S$ has in-degree 0 and out-degree 1. 
   \item
   Each sink $z \in Z$ has in-degree 1 and out-degree 0 and $\vert Z \vert = r$.
   \item
   If $uv$ is an edge in $E$ then $vu$ is not an edge in $E$. 
\end{itemize}

We assume without loss of generality that all color classes have full rank, that is $\vert S \vert = r k$. This can be assumed by adding dummy sources to $S$.

\begin{definition}
\begin{itemize}
    \item 
    A \textbf{feasible flow} in a digraph $D$ from a collection $S' \subset S$ 
    is a collection of paths $\{p^s \mid s \in S'\}$
    such that (1) $p^s$ is a simple path from $s$ to some
    sink, and (2) no vertex or edge in $D$ has more
    than $k$ such paths passing through it. 
      \item 
    A \textbf{feasible routing} in a digraph $D$ from a collection $S' \subset S$ 
    is a collection of paths $\{p^s \mid s \in S'\}$
    such that (1) $p^s$ is a simple path from $s$ to some
    sink, and (2) no vertex or edge in $D$ has more
    than one such path passing through it. 
\end{itemize}
\end{definition}

\section{The Partition Reduction Algorithm}

This section gives the Partition Reduction Algorithm. First, we define a corresponding flow graph. Using a Cycle-Canceling Algorithm, we decompose the flow graph into a collection of trees. Then we algorithmically create the partitions from the local structure in these trees. The analysis of the algorithm is deferred to the next section.

\subsection{Defining the Flow Graph}

Given the digraph $D$ we can compute a minimum
$k$ such that $M$ is $k$-colorable 
in polynomial time using a polynomial-time
algorithm for matroid intersection~\cite{oxley2006matroid}.
Further we can compute the collection of
resulting color classes $\cC = \{C_1, C_2, \ldots, C_k \}$. So $\cC$ is a partition of the sources $S$,
and for each $C_i \in \cC$  there exist $r$ vertex-disjoint paths $p_i^1, \ldots, p_i^r$ in the digraph $D$ from the $r$ sources $C_i$ to $Z$. 
We create an $f$ where the flow $f(u, v)$ on each 
edge $(u, v)$ is initialized to the number
of paths $p_i^j$ that traverse $(u, v)$,
that is $$f(u, v) = \sum_{i=1}^k \sum_{j=1}^r \mathbbm{1}_{(u, v) \in p_i^j}$$

A flow $f$ is acyclic  if for every undirected cycle $C$ 
in $D$ at least one edge in $C$ 
either has flow $k$ 
or has no flow in $f$. An arbitrary flow can 
be converted acyclic by finding cycles in a residual network $D^r$.
This is standard~\cite{ahuja1993network}, but for completeness we define  it here.

For every directed edge $(u, v)$ with $f(u, v) < k$ there
exists a forward directed edge $(u, v)$ in $D^r$ with capacity 
$c_r(u,v):= k - f(u,v)$. 
For every directed edge $(u, v)$ with $f(u, v) > 0 $ there
exists a backward directed edge $(v, u)$ in $D^r$ with capacity 
$c_r(v,u):= f(u,v)$. 
An augmenting cycle in $D^r$ is a simple directed cycle 
with strictly more than two edges.

\medskip
\noindent \textbf{Cycle-Canceling Algorithm.}  While there exists an augmenting cycle $C$ do the following:

\begin{itemize}
\item Let $c:= \min_{(u, v) \in C} c_r(u, v)$ be the minimum capacity of an
edge in $C$. 
\item For each forward edge $(u, v) \in C$,
increase $f(u, v)$ by $c$.
\item For each backward edge $(u, v) \in C$,
decrease $f(u, v)$ by $c$.
\end{itemize}

As every iteration increases the number of edges that have flow $k$ in $f$ or
that have no flow in $f$ by 1, the Cycle-Canceling Algorithm terminates after at most $|E|$ iterations. 
The following observations are straight-forward.

\begin{observation}
The following properties hold when the Cycle-Canceling
Algorithm terminates:

\begin{itemize}
\item $f$ is a feasible flow of $kr$ units of flow
from all the sources.  
\item Every undirected cycle $C$ in $D$  contains at least one 
edge with flow $k$ in $f$ or one edge with no flow in $f$. 
\item
The collection of edges in $D$ that has flow strictly between $0$ and $k$ 
in $f$ forms a forest. 
\item
The collection of  edges  in $D$ with flow $k$ in $f$ are a disjoint union of directed paths, which we will call \textbf{highways}. 
\end{itemize}
\end{observation}

\subsection{Properties of the Acyclic Flows}

We now give several definitions and 
straightforward observations about  our acyclic flow
$f$ 
that will be useful in our algorithm design and analysis.

\begin{definition}~
\begin{itemize}
    \item A vertex $v$ is a source portal if its in-degree in $D$ is 1, and it has $k$ units of flow
    passing through it in $f$. 

       \item A vertex $v$ is a sink portal if its out-degree in $D$ is 1, and it has $k$ units of flow
    passing through it in $f$. 
    \item Let $\cT$ be the forest  consisting
    of edges in $D$ that have  flow in $f$ strictly
    between $0$ and $k$. 
    \item
    For a tree $T \in \cT$ and a vertex $v \in T$ define $T_v$ to be the forest that results
    from deleting the vertex $v$ from $T$.
\end{itemize}
\end{definition}

\begin{observation}
Each sink $z \in Z$ is in a tree $T \in \cT$ that consists solely of $z$.
\end{observation}
\begin{proof}
By assumption, the sink $z$ has in-degree 1 in $D$ and all color classes $\cC$ have full rank. Hence, $k$ units of flow are entering $z$ through a unique edge.
\end{proof}

As our partition reduction algorithm partitions each tree $T \in \cT$
independently, it will be notationally more convenient
to fix an arbitrary
tree $T \in \cT$, and make some definitions relative to this fixed 
$T$, and make some 
observations that must hold for any such $T$. 
To a large extent these observations are intended to
show that the 
Figure \ref{fig:backbone} is accurate.

\begin{definition}~
\label{definition:portals}
\begin{itemize} 
    \item
    Let $\tilde{S}$ be the collection of source portals in tree $T$.
      \item
    Let $\tilde{Z}$ be the collection of sink portals in tree $T$.
    \item
    A normal vertex is  a vertex that is none of a source, a sink,
    a source portal, nor a sink portal.
\end{itemize}
\end{definition}

\begin{definition}~
\begin{itemize}
    \item 
    A \textbf{feasible flow} in $T$ from a collection $S' \subset S$ 
    is a collection of paths, one path $p^s$ 
    for each $s \in S'$ and $k$ paths $p^{\tilde{s}}_1, \ldots, p^{\tilde s}_k$ for each source portal $\tilde{s} \in \tilde{S}$ 
    such that (1) $p^s$ is a simple path from $s$ to some
    sink portal, (2) each $p^{\tilde{s}}_i$ is a simple path 
    from $\tilde s$ to a sink portal, and
    (3) no vertex or edge in $T$ has more
    than $k$ such paths passing through it. 
      \item 
     A \textbf{feasible routing} in $T$ from a collection $S' \subset S$ 
    is a collection of paths, one path $p^s$ 
    for each $s \in S'$ and one path $p^{\tilde{s}}$ for each source portal $\tilde{s} \in \tilde{S}$ 
    such that (1) $p^s$ is a simple path from $s$ to some
    sink portal, (2) $p^{\tilde{s}}$ is a simple path 
    from $\tilde s$ to a sink portal, and
    (3) no vertex or edge in $T$ has more
    than one such paths passing through it. 
\end{itemize}
\end{definition}

This following observation holds for trees in $\cT$ initially and gives intuition for the structure of $\cT$.  We remark that this observation may not hold throughout the execution of our algorithm for all trees. 

\begin{observation} \label{obs:exactk}
The number of sources in $T$ is an integer multiple of $k$.
\end{observation}
\begin{proof}
This follows from the fact that each source portal  $\tilde{s} \in T$ has exactly $k$ units of flow coming into $T$ via
$\tilde{s}$  in the flow
$f$
and each sink portal $\tilde{z} \in T$ has exactly $k$ units of flow leaving $T$ via $\tilde{z}$ in  $f$. 
\end{proof}

\begin{definition}~
\begin{itemize}
\item For two vertices $u,v \in T$, let $P(u,v)$ be the unique undirected path from $u$ to $v$ in $T$. 
    \item
    The \textbf{backbone} $B$ of $T$ is the subgraph of $T$ consisting of the union of all paths in between pairs of sink portals  in $T$,
    that is $B = \bigcup_{\tilde{y} \in \tilde{Z}} \bigcup_{ \tilde{z}  \in \tilde{Z}} 
    P(\tilde{y}, \tilde{z})$. 
      \item
    For the backbone $B$,  let $B_v$ be the induced forest that results from
     deleting $v$ from $B$. 
    \item
 A vertex $v$ in a backbone $B$ is a \textbf{branching vertex} if either:
 \begin{itemize}
     \item $v$ is not a sink portal and the forest $B_v$ contains at least two trees that each contain exactly one sink portal, or
     \item
     $v$ is a sink portal and the forest $B_v$ contains 
      at least one tree that contains exactly one sink portal.
     \end{itemize}
     \item
     Let  $\cH$ be the forest   that results from deleting 
the edges in $B$ from the tree $T$.
\item 
For two vertices $u, v \in B$, let $S(P(u, v))$ be the
sources $s \in S$ such that there exists a tree
$H \in \cH$ such that $s \in H$ and such that $H$ contains
a vertex $w \in P(u, v)$. Intuitively these are the sources
in trees in $\cH$ hanging off vertices of the path $P(u, v)$. 
\item Let $S(v)$ denote $S(P(v,v))$. 
\end{itemize}
\end{definition}

\begin{figure}[t] 
 \centering
     \includegraphics[width=0.8\textwidth]{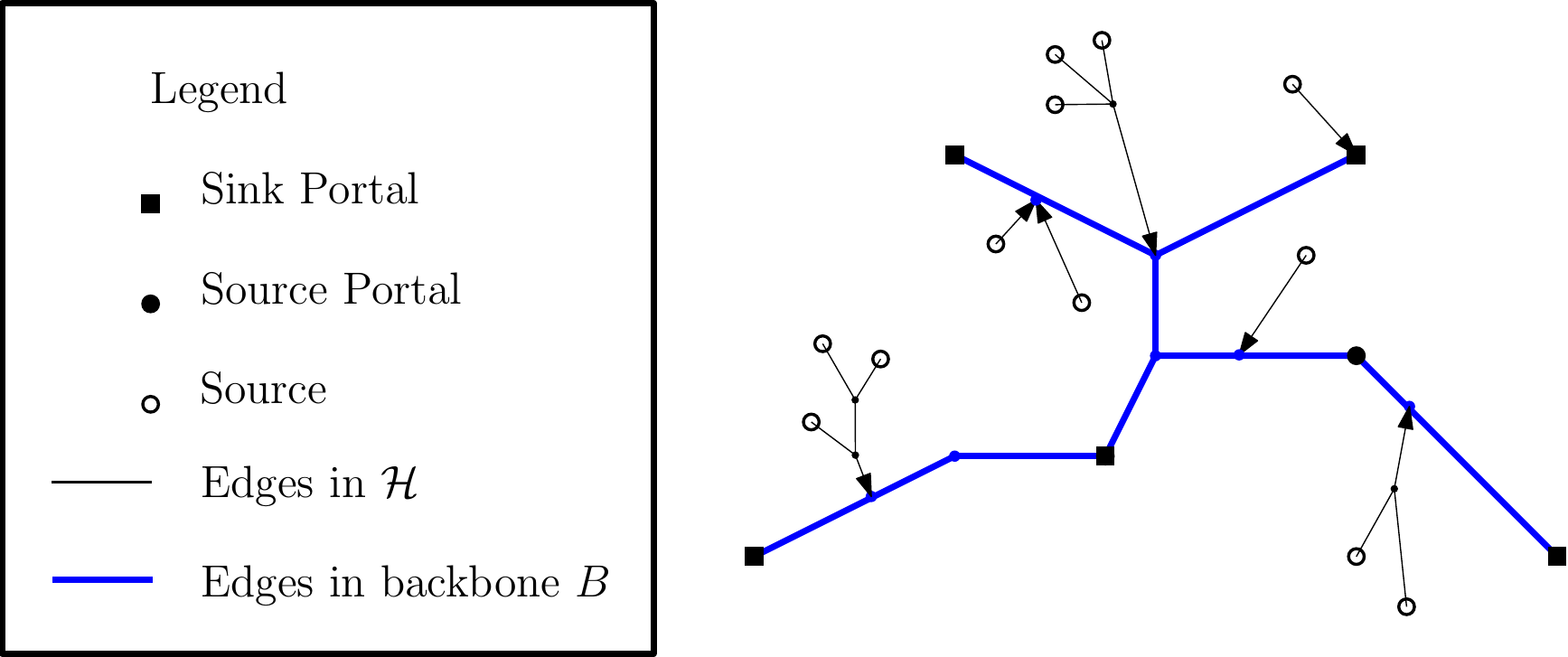}
     \caption{Backbone of a tree.}
     \label{fig:backbone}
\end{figure}

\begin{observation}\label{lem:sourcep}
If $\tilde{s} \in \tilde{S}$ is  a source portal in $T$
then $\tilde{s}$ is in the backbone $B$ and $\deg^+_B(\tilde{s})\geq 2$,
that is $\tilde{s}$ has out-degree at least 2 in $B$. 
\end{observation}
\begin{proof}
 By definition $\tilde{s}$ has a unique incoming edge,
 which is saturated in $f$, at least one outgoing
 edge in $T$ that is not saturated in $f$.
 Hence, $\deg_B^+(\tilde{s})\geq 2$. By flow conservation, there has to be at least two directed paths from $\tilde{s}$ to two different sink  portals in $T$. This implies that $\tilde{s}$ is in the backbone $B$.
\end{proof}

\begin{observation}
\label{obs:existsbranching}
If $B$ contains at least two sink portals, then $B$ contains
a branching vertex $v$. 
\end{observation}
\begin{proof}
Consider an arbitrary vertex $v \in B$. If $v$ is not a branching
vertex, then there must be a subtree $T' \in B_v$ that contains
two sink portals. One can then recurse on $T'$ to find
a branching vertex. 
\end{proof}

\begin{observation}\label{lem:sourcepbackbone}
If $s \in S$ is  a source  in $T$
then $s$ is not in the backbone $B$. 
\end{observation}
\begin{proof}
As $s$ has out-degree 1 in $D$, it can not be on any path
between sink portals in $T$. 
\end{proof}

\begin{observation} \label{lem:souresonback}
 For each tree
$H \in \cH$ it must be the case that all edges in $H$ are directed
towards the unique vertex $w$ in $H$ that is also in $B$. 
\end{observation}
\begin{proof} 
This follows from the 
fact that $H \setminus \lbrace w \rbrace$ can not contain a sink portal. 
\end{proof}

\begin{observation}
\label{lem:structuralproperties}
Assume that $T$ has at least two sink portals.  
Let $v$ be a branching vertex. Let $T'$ be a tree in the
forest $B_v$ that contains exactly one sink portal $\tilde z$.
Then the following must hold:
\begin{itemize}
\item $T' = P(v, \tilde{z}) \setminus \{v\}$. 
\item If $T'$ contains a source portal $\tilde{s}$, 
 then $\deg^+_B(\tilde{s})=2$.
\item The path $T'$ contains at most one 
vertex $y$ such that $\deg^+_B(y)=2$. 
\end{itemize}
\end{observation}

\begin{proof}
The first statement follows from the definition
of $B$ and the fact that $T'$ only contains one sink portal. 
The second statement follows since every vertex on a path
other than its endpoints has degree two. 
For the last statement assume to reach a contradiction that
there were two such $y$'s, $y_1$ and $y_2$ with $y_1$
being closer to $v$ in $B$. Then the flow leaving $y_1$
toward $\tilde z$ could not be feasibly routed through
$y_2$.
\end{proof}

\subsection{Description of the Partition Reduction Algorithm}

Given the collection of trees $\cT$, our Partition Reduction Algorithm  returns a partition $\cX$ of the sources in $S$. The algorithm 
iterates through the trees $T$ in $\cT$ and
partitions the sources in $T$ 
based on their locality in $T$. So let us consider
a particular tree $T \in \cT$.

The algorithm performs the first listed case
below that applies, with the base cases being checked
before the other cases. In the non-base cases 
the tree $T$ will be modified, and the algorithm 
called tail-recursively on the modified tree. 
 We will show after the algorithm description 
 that the algorithm maintains the invariant that
 there is a feasible flow on the tree $T$ throughout the recursion.

\medskip
 
\noindent \textbf{Base Case A:} If $T$ contains 
no sources then the recursion terminates,
and the algorithm moves to the next tree
in $\cT$. 
 
 \noindent \textbf{Base Case B:}
 Otherwise if $T$ contains at most $2k-2 $ sources and no source portal then
 these sources are added as a part $X$ in $\cX$.
 The recursion terminates, and the algorithm then moves to the next tree in $\cT$. 
 
\medskip

We perform the following recursively on $T$ if neither base case holds. Let $v$ be an arbitrary branching vertex in 
$B$. We will show this must exist in Observation~\ref{obs:brancingexistalg}.

Let $\tilde{z}_1$ be a sink portal
in some tree $T_1$ in $T_v$ that only contains
one sink portal.
If $v$ is not a sink portal, 
let $\tilde{z}_2$ be a sink portal
in some tree $T_2$, where $T_1 \ne T_2$, in $T_v$ that only contains
one sink portal.
If $v$ is a sink portal let $\tilde{z}_2 = v$.

The algorithm's cases are broken up as follows.  Case 1 is executed when there is a
source portal at $v$ or in $T_1$ or in $T_2$. 
Case 2 is executed when there is a vertex of out-degree 2 in $T_1$ or $T_2$ and there is no source portal.  Case 3 is everything else.

\medskip
 
\noindent \textbf{Recursive Case 1a:}  The vertex
$v$ is a source portal.
In this case  $T$ is modified as follows:
(1) for each source $s \in T_1$ a
directed edge $(s, v)$ is added to $T$, (2) $v$ is converted
into a normal vertex, and (3) all the nonsources in $T_1$ 
are deleted from $T$. 
The algorithm then recurses on this new $T$.

\noindent \textbf{Recursive Case 1b:}  In this case for
some $i \in \lbrace 1, 2\rbrace$ the path
$P(v, \tilde{z_i})\setminus \lbrace v \rbrace$  contains a source portal.
 In this case
$T$ is modified as follows:
(1) for each source $s \in T_i$ a
directed edge $(s, v)$ is added to $T$, 
 and (2) all the nonsources in $T_i$ 
are deleted from $T$. 
The algorithm then recurses on this new $T$.

\begin{figure}[H] 
 \centering
     \includegraphics[width=0.9\textwidth]{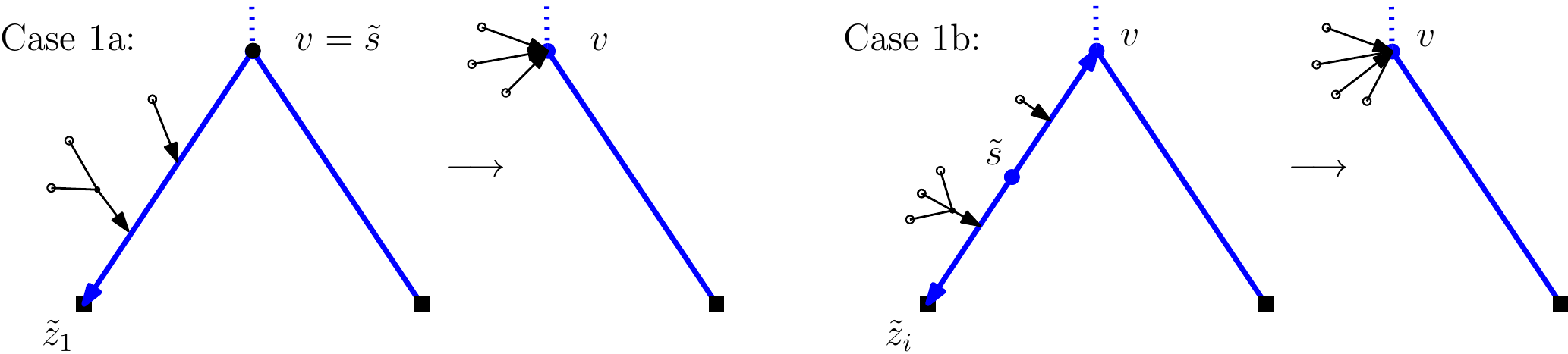}
\end{figure}

\noindent \textbf{Recursive Case 2a:} In this case for some $i \in \lbrace 1, 2\rbrace$, the path $P(v,\skp_i)\setminus \lbrace v \rbrace$ contains a vertex $y$ with $\deg^+_B(y)=2$ and $\vert S(P(y,\skp_i)) \vert \leq 2k-2$. Add the sources in $S(P(y,\skp_i)) $ as a part $X$ to $\cX$. 
The tree $T$ is then modified as follows: 
(1) for each source $s \in T_i - X$ a
directed edge $(s, v)$ is added to $T$, 
 and (2) the sources in $X$ and all the nonsources in $T_i$ 
are deleted from $T$. 
The algorithm then recurses on this new $T$.

\noindent\textbf{Recursive Case 2b:} In this case for some $i \in \lbrace 1, 2\rbrace$, the path $P(v,\skp_i)\setminus \lbrace v \rbrace$ contains a vertex $y$ with $\deg^+_B(y)=2$ and $\vert S(y) \vert = k$.
In this case the algorithm adds the $k$ sources in $S(y)$ as a part $X$ to $\cX$. 
The tree $T$ is then modified as follows: 
(1) for each source $s \in T_i - X$ a
directed edge $(s, v)$ is added to $T$, 
 and (2) the sources in $X$ and all the nonsources in $T_i$ 
are deleted from $T$. 
The algorithm then recurses on this new $T$.

\begin{figure}[H] 
 \centering
     \includegraphics[width=0.9\textwidth]{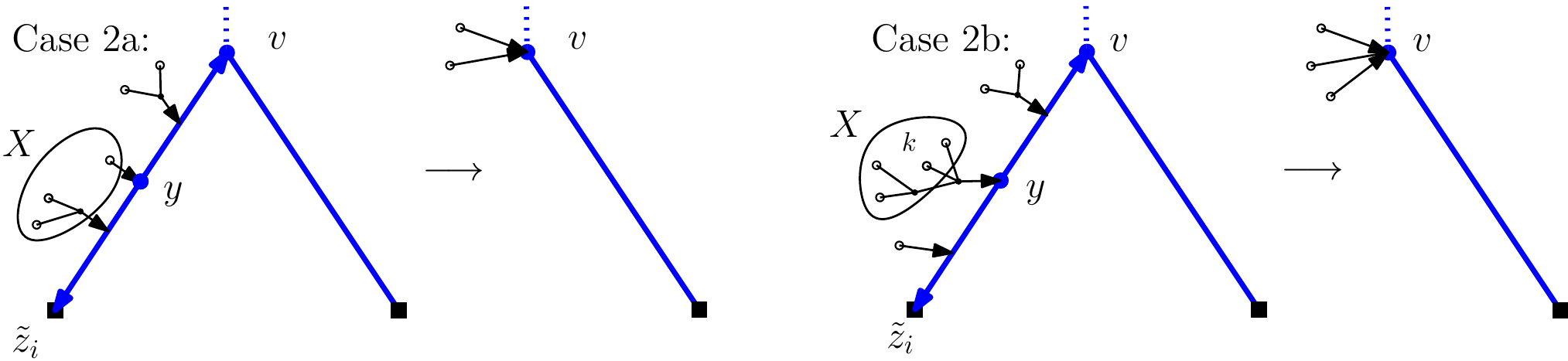}
\end{figure}

\noindent\textbf{Recursive Case 2c:} In this case for some $i \in \lbrace 1, 2\rbrace$, the path $P(v,\skp_i)\setminus \lbrace v \rbrace$ contains a vertex $y$ with $\deg^+_B(y)=2$ and $\vert S(P(y,\skp_i)\setminus \lbrace y \rbrace) \vert = k$.
In this case the algorithm adds the $k$ sources in $S(P(y,\skp_i)\setminus \lbrace y \rbrace)$ as a part $X$ to $\cX$. 
The tree $T$ is then modified as follows: 
(1) for each source $s \in T_i - X$ a
directed edge $(s, v)$ is added to $T$, 
 and (2) the sources in $X$ and all the nonsources in $T_i$ 
are deleted from $T$. 
The algorithm then recurses on this new $T$.

\begin{figure}[H] 
 \centering
     \includegraphics[width=0.43\textwidth]{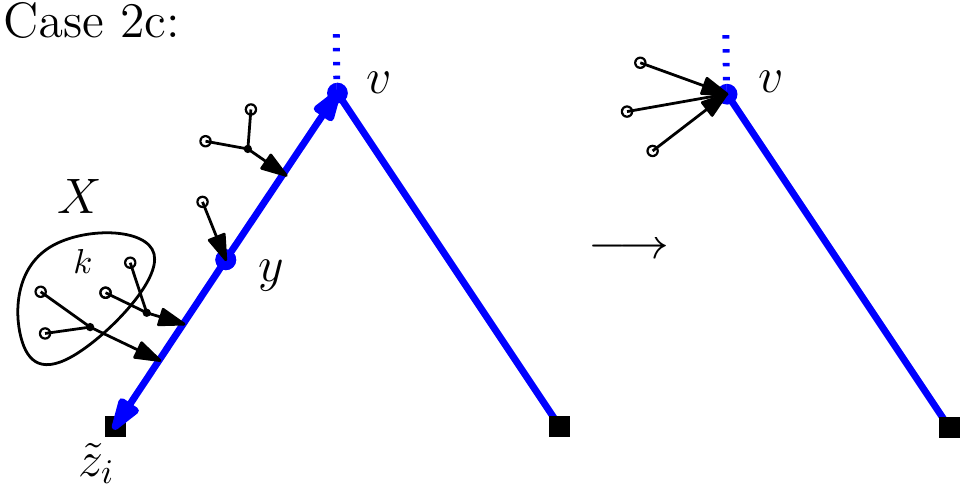}
\end{figure}

\noindent \textbf{Recursive Case 3a:} 
In this case for some $i \in \lbrace 1, 2\rbrace$ 
$T_i$ contains exactly $k$ sources. Add the sources
in $T_i$ as a part  $X$ to $\cX$. 
The tree $T$ is modified by deleting $T_i$. 
The algorithm then recurses on this new $T$.

\medskip
\noindent \textbf{Recursive Case 3b:} The set of sources in $T_1 \cup T_2$ are added as a part $X$ in $\cX$. The tree $T$ is modified by deleting the vertices in $T_1$
and $T_2$.  Add a new sink portal $\skp$ together with a directed edge $(v,\skp)$ to $T$. The algorithm then recurses on this new $T$. 

\begin{figure}[H] 
 \centering
     \includegraphics[width=0.9\textwidth]{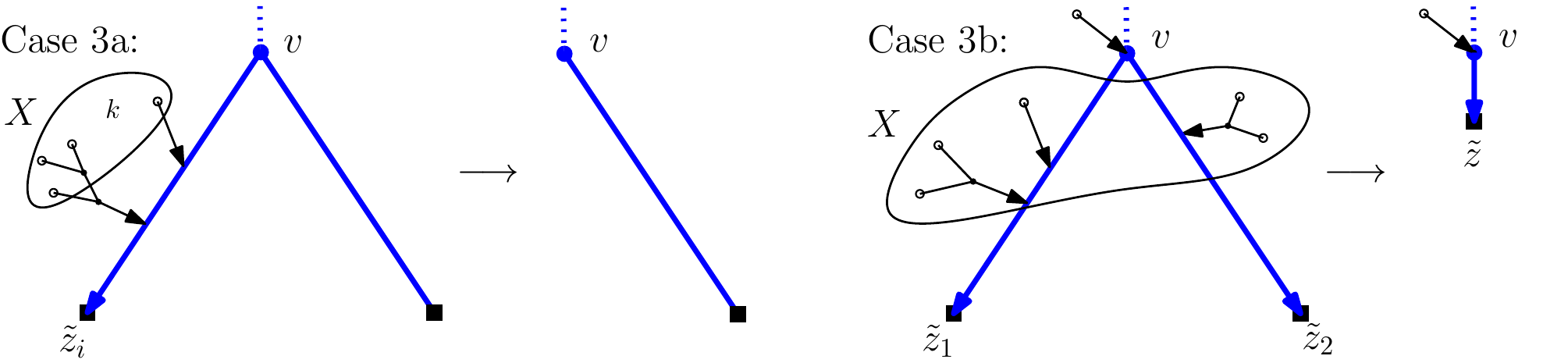}
\end{figure}

\section{Analysis of the Partition Reduction Algorithm}

Our goal is to show that the partition matroid, represented by the partition constructed from the trees, indeed corresponds to a feasible partition reduction from the gammoid. The analysis has the following key components.

\begin{itemize}
    \item Every tree $T$ has a corresponding feasible flow throughout the algorithm.
    \item Every part $X$ of the partition has size at most $2k-2$ and all sources are in some part.
    \item Any collection of sources $Y$ such that $|Y \cap X| \leq 1$ for all $X \in \cX$ is in $\cI$ and, therefore, can each route a unit of flow to the sink in $D$.
\end{itemize}

\subsection{The Trees Always Have a Feasible Flow}

This section's goal is to show that each tree has a feasible flow as defined in Definition  \ref{definition:portals} throughout the execution of the algorithm.   We will  later use this to prove that our partition indeed represents a partition matroid that corresponds to a feasible partition reduction in the following section.

We begin by  showing various invariants hold for each tree when a feasible flow exists.   In particular, this will show that a branching vertex exists if any of the recursive cases are executed. Moreover, arriving at Cases (3a) and (3b) ensure the existence of $T_1$ and $T_2$.  All together this with the fact that each tree has a feasible flow will establish that the algorithm always has a case to execute if $\cT$ is non-empty.    

This observation shows a branching vertex exists if neither base case holds. 

\begin{observation} \label{obs:brancingexistalg}
Fix a tree $T \in \cT$ during the execution of the algorithm and say $T$ supports a feasible flow as defined in Definition~\ref{definition:portals}.  If neither of the base cases apply then $T$ contains at least
two sink portals. Moreover a branching vertex must exist in $T$ in this case.
\end{observation}
\begin{proof}
This observation holds because $T$ must contain either more than $2k-2$ sources or a source portal along with at least one source.  In either case, we require two sink portals to support the strictly more than $k$ units of flow from these sources and source portal. A branching vertex must then exist by Observation \ref{obs:existsbranching}.
\end{proof}

\begin{observation}\label{obs:vtoskp}
Fix a tree $T \in \cT$ during execution of the algorithm and say $T$ supports a feasible flow as defined in Definition~\ref{definition:portals}. Say that $T$ has a branching vertex $v$ with a tree $T_i$ containing exactly one sink $\skp_i$. Moreover say that there is no vertex with out-degree 2 in $P(v,\skp_i)$. It is the case that $P(v,\skp_i)$ is a directed path from $v$ to $\skp_i$.
\end{observation}
\begin{proof}
No vertex with out-degree 2 is in $P(v,\skp_i)$.  Thus, $P(v,\skp_i)$ is either a path from $v$ to $\skp_i$ or from $\skp_i$ to $v$. Sink portals always have out-degree $0$ in $T$, so the observation follows. We note that, sink portals have out-degree 0 in $T$ initially and are never given outgoing edges by the algorithm. 
\end{proof}
 The next observation shows that a branching vertex is not a sink portal when Cases (3a) or (3b) are executed.

\begin{observation}  \label{obs:non-sink}Fix a tree $T \in \cT$ during execution of the algorithm and say $T$ supports a feasible flow as defined in Definition~\ref{definition:portals}. Say that $T$ has a branching vertex $v$ and a corresponding tree $T_i$ with exactly one sink $\skp_i$.  If $P(v,\skp_i)\setminus \lbrace v \rbrace$ does not contain a vertex of out-degree 2 in $B$, then $v$ is not a sink portal. 
\end{observation}
\begin{proof} Observation~\ref{obs:vtoskp} implies that $P(v,\skp_i)$ is a directed path form $v$ to $\skp_i$.    Sink portals always have out-degree $0$ in $T$, so the observation follows.  We note that, sink portals have out-degree 0 in $T$ initially and are never given outgoing edges by the algorithm.
\end{proof}

The previous observations guarantee the algorithm always has a case to execute if a feasible flow exists in all trees.  The next lemma guarantees the existence of a feasible flow.

\begin{lemma} Fix any tree $T$ during the execution of the algorithm.   There must be a feasible flow in $T$ as described in Definition~\ref{definition:portals}. \label{lem:feasibleflow}
\end{lemma}

\begin{proof}

The statement trivially holds at the beginning of the algorithm. In particular, each edge of $D$ supports at most  $k$ units of flow. Initially, every node in each tree $T$  has out-degree or in-degree 1 ensuring that no vertex supports more than $k$ units of flow.  We remark that during the execution of the algorithm it can be the case that there are vertices with both out- and in-degree more than one in some trees.

We inductively show that it holds after a single iteration.  Let $T^b$ be the tree at the beginning of the iteration and $T^e$ the tree at the end.  Let $f^b$ be the initial flow on $T^b$ that is feasible and we will use this to construct a flow $f^e$ on $T^e$.  We break the proof into the case that gets executed. 

Consider Recursive Case (1a). Let $S'$ be the set of sources from the tree $T_1$ which have newly created edges into $v$. There must be at least $|S'|$ units of flow in $f^b$ departing $v$ that (1) do not enter $T_1$ and (2) originate from sources in $T_1$ or the source portal $v$.  This is because $T_1$ has one sink $\skp_1$ that can support at most $k$ units of flow.  Let $F = \{f_1, f_2, \ldots \}$ denote these units of flow from $v$ to sinks. Set $f^e$ to include all flow paths in $f^b$ except those in $F$. Note these paths never enter the deleted nodes or edges.  Next observe that the sink $\skp_1$ is removed and $v$ is no longer a source portal.  Thus, we simply need to find a way to send flow for the sources in $S'$.  These nodes each route one unit of flow to their new edge directly to $v$ and then choose a unique flow path in $F$ and add these paths to $f^e$.  Notice that no vertex or edge capacity is violated by definition of $f^b$ and $F$.

Consider Recursive Cases (1b), (2b) and (2c). Again, let $S'$ be the set of sources which have newly created edges into $v$. Let $T_i$ be the tree that initially contained the sources.  In this case there must be exactly $|S'|$ units of flow in $f^b$ entering $v$ from nodes in $T_i$.  This is because in Case (1b) $T_i$  contains one source portal, one sink portal and there are $|S'|$  sources in $T_i$. In Case (2b), $T_i$  contains one sink portal, $|S(y)|=k$ and there are $|S'|$  other sources in $T_i$.  In Case (2c), $T_i$  contains one sink portal, $|S(P(y,\skp_i)\setminus \{y\})|=k$ and there are $|S'|$  other sources in $T_i$. Where this flow goes in $f^b$ will exactly correspond to where the flow in $f^e$ goes for the sources that now connect to $v$.  That is, $f^e$ has  the same flow on every edge as $f^b$ for the edges shared by $T^b$ and $T^e$. Additionally, there is one unit of flow from each of the $|S'|$ sources on the new edge that connects to $v$.

Consider Recursive Case (2a).  In this case the node $y$ has out-degree two and the sources $S(P(v, y) \setminus \{v,y\})$  have a new direct edge into $v$.  Notice that any node in $S(P(v, y) \setminus \{v,y\})$ cannot reach $\skp_i$ (since they can't route through $y$ to $\skp_i$).  Thus, these sources route to a sink through $v$. In this case, $f^e$ is the same as $f^b$ except these sources now  directly send their flow to $v$ and then follow their remaining path to a sink as in $f^b$.

Consider Recursive Case (3a).  Let $T_i$  be the tree with exactly $k$ sources.  There is one sink portal in $T_i$.  No flow in $f^b$ can enter $T_i$ via $v$.  This is because then $v$ has an outgoing edge in $T$ to a node in $T_i$ by Observation~\ref{obs:vtoskp}.  Then all $k$ sources in $T_i$ and at least one unit of flow from $v$ must route to $\skp_1$, contradicting that $\skp_1$ receives at most $k$ units of flow in $f^b$.  Set $f^e$ to be the same as $f^b$ for all shared edges between $T^b$ and $T^e$, except remove any flow sent through $v$ from sources in $T_i$ in $f^b$.

Consider Recursive Case (3b). We know $v$ is not a sink portal by Observation~\ref{obs:non-sink}. Therefore $T_1$ and $T_2$ exist, each with a unique sink portal.  At most $k$ units of flow in $f^b$ can route through $v$ by definition of $f^b$.  The flow $f^e$ is set to the same as $f^b$ on the shared edges between $T^b$ and $T^e$. Further, the flow from $v$ to sinks $\skp_1$ and $\skp_2$ in $f^b$ now go directly to the new sink added that is adjacent to $v$ in $T^e$.  This sink must receive at most $k$ units of flow.

\end{proof}

\subsection{Bounding the Size of the Parts in the Partition}

This section shows that every source is in some part $X$ in $\cX$ and that every $X \in \cX$ has size at most $2k-2$.  Thus we have a valid partition with each part having the desired size. 

\begin{lemma}
It is the case that $|X| \leq 2k -2$ for all $X \in \cX$.  Moreover, every source in $S$ is in some set in $\cX$.
\end{lemma}

\begin{proof}
It is easy to see that every source in $S$ is in some set in $\cX$.  This is because sources are always contained in some tree of $\cT$ until they are added to a set placed in $\cX$ and the algorithm stops once there is no tree in $\cT$.

Now we show how to bound the size of sets in $\cX$.  Cases (1a) and (1b) do not add a set to $\cX$.  Cases  (2b), (2c) and (3a) add a set to $\cX$ of size $k$ by definition.  Case (2a) adds a set of size $2k-2$. 

Case (3b) is more interesting. Consider the execution of this case on a tree $T$ with branching vertex $v$. Let $\skp_1$ and $\skp_2$ be the corresponding sinks. We know $v$ is not a sink portal by Observation~\ref{obs:non-sink} and therefore $T_1$ and $T_2$ both exist. By Observation \ref{obs:vtoskp}, the paths from $v$ to $\skp_1$ and $\skp_2$ are directed paths from $v$ to $\skp_1$ and from $v$ to $\skp_2$. Hence, neither $T_1$ nor $T_2$ contains more than $k$ sources. Since case (3a) does not hold, $T_1$ and $T_2$ has strictly less than $k$ sources.  Thus, there are at most $2k-2$ sources in the set added to $\cX$.
\end{proof}

\subsection{Routing Sources within a Tree}

In this section, we show that the algorithm is indeed a partition reduction from a $k$-colorable gammoid to a $(2k-2)$-colorable partition matroid. That is, we show that every set $Y$, that is independent in the partition matroid represented by $\cX$ is also independent in the gammoid. More specifically, for any set $Y$ where $|Y \cap X| \leq 1$ for all $X \in \cX$ it is the case that $Y \in \cI$ or, equivalently, there is a feasible routing in the digraph $D$ from the sources in $Y$.
The key to this will be Lemma \ref{lemma:Trouting} which essentially states
that there exists a feasible routing in each tree $T \in \cT$.

\begin{lemma}
\label{lemma:Trouting} 
Let $T$ be an arbitrary tree in $\cT$. Let $\cX_T$ be the parts in $\cX$ that are also in $T$. For a set $Y$ with $\vert Y \cap X \vert \leq 1$ for all $X \in \cX$, let $Y_T$ be the subset of sources in $Y$ that are also in $T$. Then there is a feasible routing $R$ in $T$ from $Y_T$.
\end{lemma}

\begin{proof}[Proof of Lemma \ref{lemma:Trouting}]
For convenience we consider one fixed tree $T$ and drop $T$ for the notation. 
The proof is by reverse induction on the recursion in our partitioning algorithm.
So the base case of the induction will be the base cases of the partitioning algorithm.

Consider Base Case (A). In this case $T$ does not contain any sources. By Lemma \ref{lem:feasibleflow} there exists a feasible flow. Scaling by a factor of $\frac{1}{k}$ implies that we can route one unit of flow from every source portal in $\tilde{S}$ to the collection of sink portals, such that no more than one unit of flow is routed through any edge or vertex. Then there must also exist an integer flow, which equals a routing in $T$.

Consider Base Case (B). In this case $\cX$ consists of a single part $X$ and there are no source portals in $T$. Let $s$ be the unique source in $X\cap Y$, then $s$ can be routed in $T$ as $s$ is routed in the flow $f$ that is guaranteed to exist by Lemma \ref{lem:feasibleflow}.

 For the recursive cases we construct a routing $R^b$ for the state $T^b$ of the
 tree $T$ before the recursive call from the routing $R^e$ that inductively
 exists for the state $T^e$ after the recursive calls.  
 
Consider Recursive Case (1a). Let $S'$ be the collection of sources in $T_1$. Every source in $S'$ has a directed edge into $v$ in $T^e$. There can be at most one source $s' \in S' \cap Y$ since by the induction hypothesis $R^e$ is a routing and at most one path can pass through $v$. If there is no $s' \in S' \cap Y$ then in the routing $R^b$, $v$ is routed along $P(v, \tilde{z_1})$ to $\skp_1$.
Otherwise let $s'$ be the unique source in $S' \cap Y$. Let $H' \in \cH$ be the tree that contains $s'$ and let $w'$ be the unique vertex in $H' \cap P(v, \skp_1)$. Then in $R^b$, $v$ follows the path of $s'$ in $R^e$ and $s'$ is routed through $H'$ to $w'$ and from $w'$ along the path $P(w',\skp_1)$ to $\skp_1$. We have to argue that there exists a directed path from $v$ to $\skp_1$ and from $s'$ to $\skp_1$ in $T_1$. Since $v$ is a source portal it has in-degree 0 and, hence, $P(v,\skp_1)$ does not contain a vertex with out-degree 2. By Observation \ref{obs:vtoskp}, $P(v,\skp_1)$ is a directed path from $v$ to $\skp_1$. Observation \ref{lem:souresonback} states that there also exists a directed path from $s'$ to $w'$. We also have to argue that no vertex or edge in $T^b$ has more than one path passing through it. This clearly holds for the edges and vertices shared by $T^b$ and $T^e$. Since the paths in $T_1$ are disjoint, it also holds for the edges and vertices that are in $T^b$ and not in $T^e$.

Consider Recursive Case (1b). Let $S'$ be the collection of sources in $T_i$. Every source in $S'$ has a directed edge into $v$ in $T^e$. There can be at most one $s' \in S' \cap Y$, since at most one path can pass through $v$ in $R^e$. If there is no $s' \in S' \cap Y$, then in the routing $R^b$ the source portal $\tilde{s}$ is routed through $P(\tilde{s}, \skp_i)$ to $\skp_i$. Otherwise let $s'$ be the unique source in $S' \cap Y$. Let $H' \in \cH$ be the tree that contains $s'$ and let $w'$ be the unique vertex in $H' \cap P(v, \skp_i)$. 

If $w' \in P(\tilde{s},v)$ then in $R^b$ $s'$ is routed through $H'$ to $w'$ and from $w'$ along the path $P(w',v)$ to $v$, where it follows the path of $s'$ in $R^e$. The source portal $\tilde s$ is routed through $P(\tilde{s}, \skp_i)$ to $\skp_i$.

If $w' \in P( \tilde{s}, \skp_i)$ then in $R^b$ $s'$ is routed through $H'$ to $w'$ and from $w'$ along the path $P(w',\skp_i)$ to $\skp_i$. The source portal $\tilde{s}$ is routed through $P(\tilde{s}, v)$ to $v$, where it follows the path of $s'$ in $R^e$. By Observation \ref{lem:structuralproperties}, $\tilde{s}$ is the only vertex with out-degree 2 in $T_i$. Hence, the path $P(\tilde{s},v)$ is a directed path from $\tilde{s}$ to $v$ and the path $P(\tilde{s},\skp_i)$ is a directed path from $\tilde{s}$ to $\skp_i$. Together with Observation \ref{lem:souresonback}, which implies that there exists a directed path from $s'$ to $w'$, it follows that the newly created paths above are directed paths from $s'$ and $\tilde{s}$ to $\tilde{Z}$. On the edges and vertices in $T^b$ that are also in $T^e$ the routing $R^b$ equals the routing $R^e$ and in both cases, the new paths in $T_i$ are disjoint. Hence, $R^b$ fulfills the property that no more than one path passes through every vertex or edge in $T^b$.

Consider Recursive Case (2a).
Let $H \in \cH$ be the tree that contains the selected source $s \in X$ and let $w$ be the unique vertex in $H \cap P(v, \skp_i)$. Then in $R^b$, $s$ is routed through $H$ to $w$ and from $w$ along the path $P(w,\skp_i)$ to $\skp_i$. Let $S'$ be the collection of sources in $S(P(v,y)\setminus \lbrace v,y \rbrace)$.  Every source in $S'$ has a directed edge into $v$ in $T^e$. There can be at most one source $s' \in S' \cap Y$ since at most one route can pass through $v$ in $R^e$. 

If there is such a source $s'\in S' \cap Y$, let $H' \in \cH$ be the tree that contains the source $s'$ and let $w'$ be the unique vertex in $H' \cap P(y, v)$. Then in $R^b$ $s'$ is routed through $H$ to $w$, from $w$ along $P(w,v)$ to $v$ where it follows the path of $s'$ in $R^e$.

Consider Recursive Case (2b).
Let $H \in \cH$ be the tree that contains the unique source in $s \in X \cap Y$. Let $S'$ be the collection of sources in $S(P(v,\skp_i)\setminus \lbrace y \rbrace)$.  Every source in $S'$ has a directed edge into $v$ in $T^e$. There can be at most one source $s' \in S' \cap Y$ since at most one route can pass through $v$ in $R^e$. 
If there is no $s' \in S' \cap Y$ then in $R^b$ $s$ is routed through $H$ to $y$ and from $y$ along the path $P(y, \tilde{z_i})$ to $\tilde{z_i}$. 
Otherwise, let $H' \in \cH$ be the tree that contains the unique source $s' \in S' \cap Y$ and let $w'$ be the unique vertex in $H' \cap P(v,\skp_i)$.
If $w' \in P(y, v)\setminus \lbrace y \rbrace$, then in $R^b$ $s'$ is routed through $H'$ to $w'$, from $w'$ along $P(w',v)$ to $v$ where it follows the path of $s'$ in $R^e$. Then $s$ is again routed through $H$ to $y$ and from $y$ along the path $P(y, \tilde{z_i})$ to $\tilde{z_i}$. 

If $w' \in P(y, \skp_i)\setminus \lbrace y \rbrace$, then in $R^b$ $s$ is routed through $H$ to $y$, from $y$ along $P(y,v)$ to $v$ where it follows the path of $s'$ in $R^e$. The source $s'$ is routed through $H'$ to $w'$ and from $w'$ along $P(w',\skp_i)$ to $\skp_i$.

Consider Recursive Case (2c). 
Let $H \in \cH$ be the tree that contains the unique source in $s \in X \cap Y$ and let $w$ be the unique vertex in $H \cap P(y,\skp_i)$. Let $S'$ be the collection of sources in $S(P(v,y))$.  Every source in $S'$ has a directed edge into $v$ in $T^e$. There can be at most one source $s' \in S' \cap Y$ since at most one route can pass through $v$ in $R^e$. 
If there is no $s' \in S' \cap Y$ then in $R^b$ $s$ is routed through $H$ to $w$ and from $w$ along the path $P(w, \tilde{z_i})$ to $\tilde{z_i}$. 
Otherwise, let $H' \in \cH$ be the tree that contains the unique source $s' \in S' \cap Y$ and let $w'$ be the unique vertex in $H' \cap P(v,y)$. Then in $R^b$ $s'$ is routed through $H'$ to $w'$, from $w'$ along $P(w',v)$ to $v$ where it follows the path of $s'$ in $R^e$. $s$ is still routed to $\tilde{z_i}$.

In the Recursive Cases (2a), (2b) and (2c), by Observation \ref{lem:structuralproperties}, $y$ is the only vertex with out-degree 2 in $T_i$. Hence, the path $P(y,v)$ is a directed path from $y$ to $v$ and the path $P(y,\skp_i)$ is a directed path from $y$ to $\skp_i$. Together with Observation \ref{lem:souresonback}, which implies that there exists a directed path from $s$ to $w$ and from $s'$ to $w'$, it follows that the newly created paths above are directed paths leaving at $s$ and $s'$. In all three cases, the routing $R^b$ on the edges and vertices in $T^b$ that are also in $T^e$ equals the routing $R^e$. In every case, the new paths in $T_i$ are disjoint. Hence in the Recursive Cases (2a), (2b) and (2c), $R^b$ fulfills the property that no more than one path passes through every vertex or edge in $T^b$.

Consider Recursive Case (3a). Let $H \in \cH$ be the tree that contains the selected source $s \in X \cap Y$ and let $w$ be the unique vertex in $H  \cap P(v, \tilde{z_i})$. Then in $R^b$ $s$ is routed through $H$ to $w$ and from $w$ along the path $P(w, \tilde{z_i})$ to $\tilde{z_i}$. By Observation \ref{obs:vtoskp}, the path $P(v, \tilde{z_i})$ is a directed path from $v$ to $\skp_i$. Together with Observation \ref{lem:souresonback}, which implies that there exists a directed path from $s$ to $w$, it follows, that the path from $s$ to $\skp_i$ is a directed path. Since the routing $R^b$ equals the routing of $R^e$ on all edges and vertices shared by $T^e$ and $T^b$ and the route from $s$ to $\skp_i$ only uses edges and vertices in $T_i$, also $R^b$ fulfills the property that no more than one path passes through every vertex or edge in $T^b$.

Consider Recursive Case (3b).
Let $H \in \cH$ be the tree that contains the selected source $s \in X \cap Y$ and let $w$
be the unique vertex in $H \cap P(v, \tilde{z_i})$. Then in $R^b$ $s$ is routed through $H$ to $w$ and from $w$ along the path $P(w, \tilde{z_i})$ to $\tilde{z_i}$.
Let $j \in \{1, 2 \} \setminus \{i\}$. If there is a source $s'$ routed to $\tilde{z}$ in $R^e$, then in
$R^b$ the route from $s'$ to $v$ equals the route from $s'$ to $v$ in $R^e$. Instead of going to $\tilde z$, we continue along the path $P(v, \tilde{z_j})$ to $\tilde{z_j}$. By Observation \ref{obs:vtoskp}, for $i \in \lbrace 1,2 \rbrace$ the paths $P(v, \tilde{z_i})$ are directed paths from $v$ to $\skp_i$. Together with Observation \ref{lem:souresonback}, which implies that there exists a directed path from $s$ to $w$, it follows that the path from $s$ to $\skp_i$ and from $v$ to $\skp_j$ are directed paths. Since the routing $R^b$ equals the routing of $R^e$ on all edges and vertices shared by $T^e$ and $T^b$ and the paths from $s$ to $\skp_i$ and from $v$ to $\skp_j$ are vertex-disjoint and in $T_1 \cup T_2$, also $R^b$ fulfills the property that no more than one path passes through every vertex or edge in $T^b$.
\end{proof}

\begin{lemma}\label{lem:routing}
Let $Y \subset S$ such that for all $X \in \cX$ it is the case that
$|X \cap Y|\leq1$. Then there exists a feasible routing from $Y$ in the digraph $D$.
\end{lemma}

\begin{proof}
This follows from Lemma \ref{lemma:Trouting} and the fact there is a unique highway into each source portal in each tree and a unique highway leaving each sink portal in every tree.
\end{proof}

\bibliography{bibgammoid}

\end{document}